\documentclass{article}

\usepackage{amsmath}
\usepackage{amssymb}
\usepackage{amsthm}
\usepackage{dsfont}
\usepackage{stmaryrd}
\usepackage{xcolor}
\usepackage{tikz}
\usepackage{cleveref}
\usepackage{fullpage}
\usepackage{authblk}

\usetikzlibrary{shapes.geometric} 

\newtheorem{thm}{Theorem}[section]
\newtheorem{cor}[thm]{Corollary}
\newtheorem{lem}[thm]{Lemma}

\crefname{thm}{Theorem}{Theorems}
\crefname{lem}{Lemma}{Lemmas}
\crefname{cor}{Corollary}{Corollaries}
\crefname{prop}{Proposition}{Propositions}
\crefname{defi}{Definition}{Definitions}
\crefname{remark}{Remark}{Remarks}
\crefname{example}{Example}{Examples}

\crefname{assumption}{Assumption}{Assumptions}
\crefname{figure}{Fig{.}}{Figs{.}}
\crefname{table}{Table}{Tables}
\crefname{problem}{Problem}{Problems}
\crefname{fact}{Fact}{Facts}
\crefname{conjecture}{Conjecture}{Conjectures}
\crefname{condition}{Condition}{Conditions}
\crefname{requirement}{Requirement}{Requirements}

\newcommand*{\emptygraph}{\emptyset}

\newcommand*{\defeq}{\stackrel{\mathrm{def}}{=}}

\newcommand*{\lfalse}{\mathbf{false}}
\newcommand*{\ltrue}{\mathbf{true}}

\newcommand*{\sfalse}{\textit{false}}
\newcommand*{\strue}{\textit{true}}

\newcommand*{\db}[1]{\left\llbracket #1 \right\rrbracket}

\newcommand*{\Expr}{\mathsf{Expr}}
\newcommand*{\Assn}{\mathsf{Assn}}

\newcommand*{\FinMea}{\mathrm{FinMea}}
\newcommand*{\support}{\mathrm{supp}}

\tikzset{
    inline/.style={scale=0.15},
    vertex/.style={circle, fill=black, inner sep=0pt, minimum size=3pt}
}

\newcommand{\Ithree}{%
  \begin{tikzpicture}[inline]
    \coordinate (a) at (90:1);
    \coordinate (b) at (210:1);
    \coordinate (c) at (330:1);
    \node[vertex] at (a) {};
    \node[vertex] at (b) {};
    \node[vertex] at (c) {};
  \end{tikzpicture}%
}

\newcommand{\Ethree}{%
  \begin{tikzpicture}[inline]
    \coordinate (a) at (90:1);
    \coordinate (b) at (210:1);
    \coordinate (c) at (330:1);
    \draw (b) -- (c); 
    \node[vertex] at (a) {};
    \node[vertex] at (b) {};
    \node[vertex] at (c) {};
  \end{tikzpicture}%
}

\newcommand{\Pthree}{%
  \begin{tikzpicture}[inline]
    \coordinate (a) at (90:1);
    \coordinate (b) at (210:1);
    \coordinate (c) at (330:1);
    \draw (b) -- (a) -- (c); 
    \node[vertex] at (a) {};
    \node[vertex] at (b) {};
    \node[vertex] at (c) {};
  \end{tikzpicture}%
}

\newcommand{\Kthree}{%
  \begin{tikzpicture}[inline]
    \coordinate (a) at (90:1);
    \coordinate (b) at (210:1);
    \coordinate (c) at (330:1);
    \draw (b) -- (a) -- (c) -- (b); 
    \node[vertex] at (a) {};
    \node[vertex] at (b) {};
    \node[vertex] at (c) {};
  \end{tikzpicture}%
}

\newcommand{\LItwo}{%
  \begin{tikzpicture}[inline]
    \coordinate (a) at (90:1);
    \coordinate (b) at (330:1);
    \node[vertex,fill=white,circle,draw] at (a) {};
    \node[vertex] at (b) {};
  \end{tikzpicture}%
}

\newcommand{\LEtwo}{%
  \begin{tikzpicture}[inline]
    \coordinate (a) at (90:1);
    \coordinate (b) at (330:1);
    \draw (a) -- (b); 
    \node[vertex,fill=white,circle,draw] at (a) {};
    \node[vertex] at (b) {};
  \end{tikzpicture}%
}

\newcommand{\LIthree}{%
  \begin{tikzpicture}[inline]
    \coordinate (a) at (90:1);
    \coordinate (b) at (210:1);
    \coordinate (c) at (330:1);
    \node[vertex,fill=white,circle,draw] at (a) {};
    \node[vertex] at (b) {};
    \node[vertex] at (c) {};
  \end{tikzpicture}%
}

\newcommand{\LEthreeC}{%
  \begin{tikzpicture}[inline]
    \coordinate (a) at (90:1);
    \coordinate (b) at (210:1);
    \coordinate (c) at (330:1);
    \draw (b) -- (c); 
    \node[vertex,fill=white,circle,draw] at (a) {};
    \node[vertex] at (b) {};
    \node[vertex] at (c) {};
  \end{tikzpicture}%
}

\newcommand{\LEthreeB}{%
  \begin{tikzpicture}[inline]
    \coordinate (a) at (90:1);
    \coordinate (b) at (210:1);
    \coordinate (c) at (330:1);
    \draw (a) -- (c); 
    \node[vertex,fill=white,circle,draw] at (a) {};
    \node[vertex] at (b) {};
    \node[vertex] at (c) {};
  \end{tikzpicture}%
}

\newcommand{\LPthreeC}{%
  \begin{tikzpicture}[inline]
    \coordinate (a) at (90:1);
    \coordinate (b) at (210:1);
    \coordinate (c) at (330:1);
    \draw (b) -- (a) -- (c); 
    \node[vertex,fill=white,circle,draw] at (a) {};
    \node[vertex] at (b) {};
    \node[vertex] at (c) {};
  \end{tikzpicture}%
}

\newcommand{\LPthreeB}{%
  \begin{tikzpicture}[inline]
    \coordinate (a) at (90:1);
    \coordinate (b) at (210:1);
    \coordinate (c) at (330:1);
    \draw (a) -- (c) -- (b); 
    \node[vertex,fill=white,circle,draw] at (a) {};
    \node[vertex] at (b) {};
    \node[vertex] at (c) {};
  \end{tikzpicture}%
}

\newcommand{\LKthree}{%
  \begin{tikzpicture}[inline]
    \coordinate (a) at (90:1);
    \coordinate (b) at (210:1);
    \coordinate (c) at (330:1);
    \draw (b) -- (a) -- (c) -- (b); 
    \node[vertex,fill=white,circle,draw] at (a) {};
    \node[vertex] at (b) {};
    \node[vertex] at (c) {};
  \end{tikzpicture}%
}

\usepackage{microtype}

\title{An Introduction to Razborov's Flag Algebra as a Proof System for Extremal Graph Theory}
\date{}

\author[1]{Gyeongwon Jeong}
\affil[1]{School of Computing, KAIST, South Korea,
  \texttt{jgyw0910@kaist.ac.kr}}

\author[2]{Seonghun Park}
\affil[2]{School of Computing, KAIST, South Korea,
  \texttt{hun57@kaist.ac.kr}}

\author[3]{Hongseok Yang}
\affil[3]{School of Computing, KAIST, South Korea,
  \texttt{hongseok00@gmail.com}}

\begin{document}
\maketitle

\begin{abstract}
Razborov's flag algebra forms a powerful framework for deriving asymptotic inequalities between induced subgraph densities, underpinning many advances in extremal graph theory. This survey introduces flag algebra to computer scientists working in logic, programming languages, automated verification, and formal methods. We take a logical perspective on flag algebra and present it in terms of syntax, semantics, and proof strategies, in a style closer to formal logic. One popular proof strategy derives valid inequalities by first proving inequalities in a labelled variant of flag algebra and then transferring them to the original unlabelled setting using the so-called downward operator. We explain this strategy in detail and highlight that its transfer mechanism relies on the notion of what we call an adjoint pair, reminiscent of Galois connections and categorical adjunctions, which appear frequently in work on automated verification and programming languages. Along the way, we work through representative examples, including Mantel's theorem and Goodman's bound on Ramsey multiplicity, to illustrate how mathematical arguments can be carried out symbolically in the flag algebra framework.
\end{abstract}

\section{Introduction}

Razborov's flag algebra~\cite{Razborov2007} is a framework for reasoning about
\emph{asymptotic} (i.e., large-size) properties of combinatorial structures.
In extremal graph theory, many problems ask for the maximum or minimum possible density
of a fixed finite pattern (such as an edge, a triangle, or an induced cycle) among all large
graphs satisfying certain constraints (such as forbidding a subgraph).  Flag algebra turns
such questions into optimisation problems over a compact space of limit objects, and provides
a symbolic calculus for deriving valid inequalities between densities of those finite patterns.

Since its introduction, the flag algebra framework has led to solutions of, and substantial progress on,
several central problems in extremal graph theory, including the minimum triangle density
problem~\cite{Razborov2008}, Erd\H{o}s's pentagon problem~\cite{HatamiHKNR13,Grzesik12},
the induced $5$-cycle density problem~\cite{BaloghHLP16}, Tur\'an's tetrahedron
problem~\cite{Razborov10,BaberT12}, and the rainbow-triangle density
problem~\cite{BaloghHLPVY17}. Many key proof steps in the framework have also been automated:
{\sc Flagmatic}~\cite{Falgas-RavryV13,PikhurkoST19} is a widely used tool that underpins
several of these results in extremal graph theory.

This article surveys flag algebra for computer scientists working in logic, programming languages, automated verification, and formal
methods. Although several excellent surveys on flag algebra exist~\cite{Razborov13,SilvaFS16,Grzesik15thesis},
they are written primarily for mathematicians,
presenting flag algebra as a symbolic calculus for familiar arguments (e.g., double counting and
applications of Cauchy--Schwarz). They also emphasise the algebraic viewpoint of flag algebra. By contrast, we take a \emph{logical} perspective and present
flag algebra in terms of syntax, semantics, and proof strategies, in a style closer to formal
logic. We also highlight that some core constructions (notably the so-called downward operator that relates
labelled and unlabelled settings) can be understood via ideas reminiscent of Galois connections
and, more generally, categorical adjunctions---concepts that frequently appear in automated verification and programming languages. The results explained in this survey article, such as lemmas and their proofs, are from Razborov's original paper on flag algebra~\cite{Razborov2007}, although they are presented here in a slightly different style to emphasise the logical perspective and fit the overall narrative.

The rest of the paper is organised as follows. In the remainder of this introduction, we fix some notation and terminology used throughout the paper. In \Cref{sec:induced-subgraph-density}, we review the notion of induced subgraph density and use it to formulate representative asymptotic problems in extremal graph theory. In \Cref{sec:limiting-graphs}, we introduce limiting graphs, which enable a concise description of these problems and form the semantic basis of flag algebra. In \Cref{sec:syntax-semantics}, we present flag algebra as a logical system with two kinds of terms---density expressions and assertions---together with a semantics in terms of limiting graphs. We also illustrate, through two examples, how standard extremal arguments can be expressed as derivations in this system. The key steps in both examples rely on a nontrivial inequality. In \Cref{sec:downward-operator}, we describe a strategy for proving such inequalities using labelled variants of flag algebra together with a mechanism for transferring inequalities from the labelled setting to the unlabelled one via the so-called downward operator. Finally, in \Cref{sec:conclusion}, we discuss research opportunities in flag algebra for computer scientists working in logic, programming languages, automated verification, and formal methods.

Before we begin, we fix some notation and terminology.
For each $n \in \mathbb{N}$, we write $[n]$ for the set $\{1,2,\ldots,n\}$.
For a set $X$ and a nonnegative integer $n$, we write $\binom{X}{n}$ for the set of all
$n$-element subsets of $X$.
Throughout the paper, we consider finite, simple, undirected graphs $G$ whose vertex set is a
subset of $\mathbb{N}$. Here, \emph{simple} means that $G$ has neither loops (i.e., edges
from a vertex to itself) nor multiple edges (i.e., more than one edge between the same pair of
vertices). Thus, such a graph $G$ is formally a pair consisting of a finite set $V \subseteq \mathbb{N}$
and a set $E \subseteq \binom{V}{2}$. We refer to these simply as \emph{graphs}, and write
$\mathcal{G}$ for the set of all graphs. Note that $\mathcal{G}$ is countably infinite.
For a graph $G$, we write $V(G)$ and $E(G)$ for its vertex set and edge set, respectively, and
$v(G) \defeq |V(G)|$ and $e(G) \defeq |E(G)|$ for their cardinalities.
For $U \subseteq V(G)$, we write $G[U]$ for the subgraph of $G$ induced by $U$, i.e., the graph with
vertex set $U$ and edge set $E(G) \cap \binom{U}{2}$.
We use standard notation for common graphs. The complete graph on vertex set $[n]$ is denoted by $K_n$
(so $E(K_n) = \binom{[n]}{2}$),
the path on $n$ vertices by $P_n$
(i.e., $V(P_n)=[n]$ and $E(P_n)=\{\{i,i+1\} : i \in [n-1]\}$),
and the edgeless graph with $n$ vertices by $I_n$ (i.e., $V(I_n) = [n]$ and $E(I_n) = \emptyset$).  

\section{Induced Subgraph Density}
\label{sec:induced-subgraph-density}

Let $H$ and $G$ be graphs with $v(H) \leq v(G)$. The central notion in this survey article is the
\emph{induced subgraph density} of $H$ in $G$, or simply the \emph{$H$-density} in $G$, defined by
\[
p(H, G) \defeq \mathbb{P}\big[G[\mathbf{U}] \simeq H\big],
\]
where $\mathbf{U}$ is a uniformly random subset of $V(G)$ of size $v(H)$, and $\simeq$ denotes graph
isomorphism. In words, $p(H,G)$ is the probability that a uniformly random set of $v(H)$ vertices of
$G$ induces a subgraph isomorphic to $H$. For instance, if $H = K_2$ and $G = K_3$, then $p(H,G) = 1$,
since every induced subgraph of $G$ on two vertices is isomorphic to $H$. As another example, if
$H = K_2$ and $G = P_3$, then $p(H,G) = \frac{2}{3}$, since among the three induced subgraphs of $G$ on
two vertices, exactly two are isomorphic to $H$.
By convention, we set $p(H,G) = 0$ if $v(H) > v(G)$.

Many asymptotic extremal problems can be phrased in terms of induced subgraph densities.
A representative example is an asymptotic induced-subgraph version of a Tur\'an-type problem: given a graph $H$
and a finite set of forbidden graphs $\mathcal{F}$ such that there exists at least one graph avoiding every graph in $\mathcal{F}$, determine the maximum possible $H$-density in large
graphs that avoid every graph in $\mathcal{F}$. Formally, the problem asks for the value of
\[
T(H, \mathcal{F}) \defeq
\limsup_{n \to \infty} \max \{ p(H, G) \,:\, v(G) = n \text{ and } p(F,G) = 0 \text{ for all } F \in \mathcal{F} \}.
\]
When $H = K_2$ and $\mathcal{F} = \{K_3\}$, $T(H,\mathcal{F})$ is the asymptotic maximum edge density
of triangle-free graphs, and Mantel's theorem states that $T(H,\mathcal{F}) = 1/2$~\cite{Mantel1907}.
More generally, when $H = K_2$ and $\mathcal{F} = \{K_{r+1}\}$ for an integer $r \geq 2$, Tur\'an's
theorem gives $T(H,\mathcal{F}) = 1 - \frac{1}{r}$~\cite{Turan1941}.

Another example is the Ramsey multiplicity problem, a counting analogue of the classical Ramsey
number problem. Formally, for integers $s,t \geq 2$, the problem asks for the value of
\[
M(s,t) \defeq
\liminf_{n \to \infty} \min \{ p(K_s, G) + p(K_t, \overline{G}) \,:\, v(G) = n \},
\]
where $\overline{G}$ is the complement of $G$ (i.e., $V(\overline{G}) = V(G)$ and
$E(\overline{G}) = \binom{V(G)}{2} \setminus E(G)$). Intuitively, $G$ and $\overline{G}$ correspond
to a $2$-colouring of the edges of $K_n$ (say, blue and red), and $p(K_s, G) + p(K_t, \overline{G})$
denotes the combined density of blue copies of $K_s$ and red copies of $K_t$. It is known that
$M(3,3) = \frac{1}{4}$ by Goodman's result~\cite{Goodman1959}, but the problem remains open for many
other pairs $(s,t)$. The flag algebra framework has been used to determine exact values in some
cases (e.g., $M(3,4)$ and $M(3,5)$~\cite{Parczyk2025new}) and to obtain improved lower bounds in
others (e.g., $M(4,4)$, $M(4,5)$, and $M(5,5)$~\cite{Grzesik22tripartite,Parczyk2025new}).

In both examples, the two arguments of the density function $p$ play different roles: the second argument $G$ is the graph under study, whereas the first argument $H$ is part of the specification and denotes the pattern whose density is being measured. This convention is commonly used in extremal graph theory and is assumed in most applications of flag algebra. Accordingly, we call $G$ the \emph{host graph} and $H$ the \emph{pattern graph}. Typically, $G$ is large and $H$ is small, and $p(H,G)$ captures a local quantitative property of $G$.

Moreover, the density function $p$ induces a representation of a graph $G$ as an infinite-dimensional vector in $[0,1]^{\mathcal{G}}$ indexed by all graphs in $\mathcal{G}$. This representation is injective up to graph isomorphism: $p(\_,G)=p(\_,G')$ if and only if $G \simeq G'$. It also lets us use the product topology on $[0,1]^{\mathcal{G}}$ to study asymptotic properties of graph sequences, as we explain next.

\section{Limiting Graphs}
\label{sec:limiting-graphs}

Let $(G_n)_{n \in \mathbb{N}}$ be a sequence of graphs. The sequence is \emph{increasing} if
$v(G_n) < v(G_{n+1})$ for all $n \in \mathbb{N}$. It is \emph{convergent} if it is increasing and
$\lim_{n \to \infty} p(H, G_n)$ exists for every graph $H$. Equivalently, the sequence of vector
representations $(p(\_,G_n))_n$ mentioned at the end of the previous section converges pointwise in
$[0,1]^{\mathcal{G}}$.
Although this definition may appear restrictive---since it requires the limit to exist for every
graph $H$---convergence is in fact common: every increasing sequence of graphs has a convergent
subsequence. Indeed, the set $\mathcal{G}$ of graphs is countable\footnote{Recall that by a graph we mean a finite
simple graph $G$ with $V(G) \subseteq \mathbb{N}$.}, so the product space $[0,1]^{\mathcal{G}}$,
equipped with the product topology, is compact and metrizable. Hence it is sequentially compact,
and the sequence of vectors $(p(\_,G_n))_{n\in\mathbb{N}}$ admits a convergent subsequence.

A function $\phi : \mathcal{G} \to [0,1]$ is a \emph{limiting graph}\footnote{In~\cite{Razborov2007}, such a function is called a \emph{positive homomorphism}. We avoid this terminology because it can be opaque outside the algebraic setting, and we want to emphasise the interpretation of $\phi$ as a limit object representing an infinite graph.} if there exists a convergent sequence of graphs $(G_n)_{n \in \mathbb{N}}$ such that
\[
\phi(H) = \lim_{n \to \infty} p(H, G_n)
\qquad\text{for every graph } H \in \mathcal{G}.
\]
Let $\Phi$ denote the set of all limiting graphs. A useful mental picture (formalised, for example,
via graphons~\cite{Lovasz2012GraphonBook}) is that a limiting graph represents a random infinite undirected graph with a continuum number of vertices, and $\phi(H)$ gives the average induced subgraph density of the pattern $H$ in this infinite host graph.
Every convergent sequence $(G_n)_n$ determines a unique limiting graph $\phi$, but the same limiting
graph may arise from multiple sequences.

Using limiting graphs, we can restate the extremal problems from \Cref{sec:induced-subgraph-density} more cleanly.\footnote{The correctness of this restatement follows from Corollary 3.4 of \cite{Razborov2007} and the fact that the asymptotic induced-subgraph Tur\'an-type problem can be rephrased in terms of unconstrained optimisation.} The
asymptotic induced-subgraph Tur\'an-type problem can be formulated as the optimisation problem
\[
T(H, \mathcal{F}) \defeq \sup\ \{ \phi(H) \,:\, \phi \in \Phi \text{ and } \phi(F) = 0 \text{ for all } F \in \mathcal{F} \}.
\]
Similarly, since $p(K_t,\overline{G}) = p(I_t,G)$, the Ramsey multiplicity problem can be written as
\[
M(s,t) \defeq
\inf\ \{ \phi(K_s) + \phi(I_t) \,:\, \phi \in \Phi \}.
\]
Here, $I_t$ denotes the edgeless graph on $t$ vertices.

A main message of this survey is that flag algebra provides a logical framework for reasoning about limiting graphs in $\Phi$, and in particular for establishing bounds in asymptotic extremal problems. For example, it offers a syntax for expressing bounds such as
\[
T(H, \mathcal{F}) \leq c
\qquad\text{and}\qquad
M(s,t) \geq c'
\]
for $c,c' \in \mathbb{R}$, together with a semantics parameterised by limiting graphs in $\Phi$. In this way, such bounds become statements about limiting graphs. Flag algebra also comes with strategies for proving these bounds, some of which can be automated. In the next two sections, we present the syntax and semantics of this logic, and a popular proof strategy for it.

\section{Syntax and Semantics of Flag Algebra}
\label{sec:syntax-semantics}

The syntax of flag algebra consists of \emph{density expressions} $E$ and \emph{assertions} $A$.
Density expressions generalise pattern graphs $H$ and denote real-valued functions on limiting graphs.
Assertions express qualitative properties of limiting graphs using density expressions and logical connectives.

Formally, the syntax is defined as follows:
\begin{align*}
\text{Density Expressions}\quad E & ::= H \,\mid\, r \cdot E \,\mid\, 0 \,\mid\, E + E \,\mid\, 1 \,\mid\, E \cdot E,
\\
\text{Assertions}\quad A & ::= \lfalse \,\mid\, \ltrue \,\mid\, E \geq E \,\mid\, \neg A \,\mid\, A \lor A,
\end{align*}
where $H$ ranges over graphs and $r$ over real numbers. We denote the set of all density expressions by $\Expr$ and that of assertions by $\Assn$.
Density expressions are built from pattern graphs and the constants $0$ and $1$ using scalar multiplication, addition, and multiplication.
Assertions are constructed from the Boolean constants $\lfalse$ and $\ltrue$, comparisons between density expressions, negation, and disjunction.

Negation and disjunction have their usual meanings from classical logic, and induce other connectives (e.g., conjunction and implication) in the standard way:
\[
A \land A' \defeq \neg(\neg A \lor \neg A')
\qquad\text{and}\qquad
(A \implies A') \defeq \neg A \lor A'.
\]
Similarly, other comparison operators between density expressions (such as $=$, $>$, $\leq$, and $<$) can be defined using $\geq$ and the logical connectives.
For instance,
\[
(E_1 = E_2) \defeq (E_1 \geq E_2) \land (E_2 \geq E_1)
\quad\text{and}\quad
(E_1 < E_2) \defeq \neg (E_1 \geq E_2).
\]
Finally, we use the abbreviations $-E \defeq (-1)\cdot E$ and $E_1 - E_2 \defeq E_1 + (-E_2)$, and we identify real constants $r \in \mathbb{R}$ with the density expressions $r \cdot 1$.

The semantics of density expressions $E$ and assertions $A$ has the form
\[
\db{E} : \Phi \to \mathbb{R}
\qquad\text{and}\qquad
\db{A} : \Phi \to \mathbb{B},
\]
where $\mathbb{B} \defeq \{\strue,\sfalse\}$.
It is defined by induction on the structure of $E$ and $A$ using the ordered commutative algebra structure of $\mathbb{R}$ and the Boolean algebra structure of $\mathbb{B}$.
For a limiting graph $\phi \in \Phi$,
\begin{align*}
\db{H}\phi & \defeq \phi(H),
&
\db{r \cdot E}\phi & \defeq r \cdot \db{E}\phi,
\\
\db{0}\phi & \defeq 0,
&
\db{E_1 + E_2}\phi & \defeq \db{E_1}\phi + \db{E_2}\phi,
\\
\db{1}\phi & \defeq 1,
&
\db{E_1 \cdot E_2}\phi & \defeq \db{E_1}\phi \cdot \db{E_2}\phi,
\end{align*}
and
\begin{align*}
\db{\lfalse}\phi & \defeq \sfalse,
&
\db{\ltrue}\phi & \defeq \strue,
\\
\db{E_1 \geq E_2}\phi & \defeq
\begin{cases}
  \strue & \text{if } \db{E_1}\phi \geq \db{E_2}\phi,
  \\
  \sfalse & \text{otherwise},
\end{cases}
&
\db{\neg A}\phi & \defeq
\begin{cases}
  \strue & \text{if } \db{A}\phi = \sfalse,
  \\
  \sfalse & \text{otherwise},
\end{cases}
\\
\db{A_1 \lor A_2}\phi & \defeq
\begin{cases}
  \strue & \text{if } \db{A_1}\phi = \strue \text{ or } \db{A_2}\phi = \strue,
  \\
  \sfalse & \text{otherwise}.
\end{cases}
\end{align*}
We say that an assertion $A$ is \emph{valid} if $\db{A}\phi = \strue$ for every limiting graph $\phi \in \Phi$.
Note that the syntax and semantics of flag algebra coincide with those of arithmetic expressions and formulas in the quantifier-free fragment of standard first-order logic over the reals. Graphs $H$ in density expressions play the role of real-valued variables, density expressions are arithmetic expressions over these variables, and assertions are formulas over these arithmetic expressions. As a result, many standard reasoning principles from first-order logic apply to flag algebra directly.
For instance, the following assertions hold for all limiting graphs $\phi \in \Phi$:
\begin{align*}
  & \Big(H_1 \cdot (H_3 + H_2) + (4 \cdot H_3) \cdot H_1\Big) = \Big(H_1 \cdot H_2 + 5 \cdot (H_1 \cdot H_3)\Big),
  \\
  & \Big((H_1 + H_2 \geq 0) \wedge (0 \geq H_2 - H_3)\Big) \implies H_1 + H_3 \geq 0.
\end{align*}

\begin{lem}
  \label{lemma:basic-algebra}
  The rules of ordered commutative algebras are valid for density expressions in flag algebra. Also, the rules of Boolean algebras are valid assertions in flag algebra.
\end{lem}

What is special about flag algebra, compared with ordinary first-order logic over reals, is that the ``variables'' (i.e., graphs $H$) cannot be assigned arbitrary real values: their values must arise simultaneously from a single limiting graph $\phi$, and are therefore correlated.
One simple consequence is that two variables representing isomorphic graphs always have the same value.
\begin{lem}
  \label{lemma:isomorphism}
  If graphs $H$ and $H'$ are isomorphic, then the equation $H = H'$ is valid.
\end{lem}
\begin{proof}
  Let $\phi$ be an arbitrary limiting graph. Then, there exists 
  a converging sequence $(G_n)_{n \in \mathbb{N}}$ of graphs such that 
  $\lim_{n \to \infty} p(F,G_n) = \phi(F)$ for all graphs $F \in \mathcal{G}$. Since $H$ and $H'$ are isomorphic, we have $p(H,G_n) = p(H',G_n)$ for all $n \in \mathbb{N}$. Therefore, $\phi(H) = \lim_{n \to \infty} p(H,G_n) = \lim_{n \to \infty} p(H',G_n) = \phi(H')$. Since $\phi$ was arbitrary, the equation $H = H'$ is valid.
\end{proof}

The next useful consequences follow from the connection between  limiting graphs and probability theory. For each graph $H$ and limiting graph $\phi$, the value $\phi(H)$ intuitively means the probability that, if we choose a subset of $v(H)$ vertices uniformly at random in the infinite graph represented by $\phi$, the induced subgraph on the chosen vertices is isomorphic to $H$. From this intuition, we see that $\phi(H) \in [0,1]$ and that the sum of $\phi(H')$ over all non-isomorphic graphs $H'$ on a fixed vertex set $[n]$ equals $1$. For a finite subset $V$ of $\mathbb{N}$, we say that a set of graphs $\mathcal{H}_V$ is a \emph{complete set of graphs on $V$} if every graph on the vertex set $V$ is isomorphic to exactly one graph in $\mathcal{H}_V$. Using this terminology, 
we restate the above observations formally in the following lemmas.
\begin{lem}
  \label{lemma:non-negativity}
  The assertion $1 \geq H \wedge H \geq 0$ is valid for every graph $H$.
\end{lem}
\begin{proof}
  Consider an arbitrary limiting graph $\phi$. Then, there exists 
  a converging sequence $(G_n)_{n \in \mathbb{N}}$ of graphs such that 
  $\lim_{n \to \infty} p(F,G_n) = \phi(F)$ for all graphs $F \in \mathcal{G}$. Since $1 \geq p(H,G_n) \geq 0$ for all $n \in \mathbb{N}$, we have
  \begin{align*}
  & \db{1}\phi = 1 \geq \lim_{n \to \infty} p(H,G_n) = \phi(H) = \db{H}\phi, \text{ and}
  \\
  & \db{H}\phi = \phi(H) = \lim_{n \to \infty} p(H,G_n) \geq 0 = \db{0}\phi.
  \end{align*}
  Since $\phi$ was arbitrary, the assertion $1 \geq H \wedge H \geq 0$ is valid.
\end{proof}
\begin{lem}
  \label{lemma:sum-to-one}
  Let $V$ be a finite subset of $\mathbb{N}$, and let $\mathcal{H}_V$ be a complete set of graphs on $V$. Then
  \[
  \sum_{H \in \mathcal{H}_V} H = 1
  \]
  is valid.
\end{lem}
\begin{proof}
  Let $\phi$ be an arbitrary limiting graph. Then, there exists 
  a converging sequence $(G_n)_{n \in \mathbb{N}}$ of graphs such that 
  $\lim_{n \to \infty} p(F,G_n) = \phi(F)$ for all graphs $F \in \mathcal{G}$. Since $\mathcal{H}_V$ is a complete set of graphs on $V$, we have
  \[
  \sum_{H \in \mathcal{H}_V} p(H,G_n) = 1
  \]
  for all $n \in \mathbb{N}$. Therefore,
  \[
  \db{\sum_{H \in \mathcal{H}_V} H}\phi =
  \sum_{H \in \mathcal{H}_V} \phi(H) = \sum_{H \in \mathcal{H}_V} \lim_{n \to \infty} p(H,G_n) = \lim_{n \to \infty} \sum_{H \in \mathcal{H}_V} p(H,G_n) = 1 = \db{1}\phi.
  \]
  Since $\phi$ was arbitrary, the equation $\sum_{H \in \mathcal{H}_V} H = 1$ is valid.
\end{proof}

We also note two further consequences that are heavily used in applications of flag algebra.
They say that the density expression $G$ for a graph $G$ can be written as a linear combination of larger graphs, and that the multiplicative unit $1$ and the multiplication in density expressions can be eliminated.
In the lemmas below, let $V$ be a finite subset of $\mathbb{N}$, and let $\mathcal{H}_V$ be a complete set of graphs on $V$.
\begin{lem}
  \label{lemma:decomposition}
  For every graph $H$ with $v(H) \leq |V|$,
  \[
  H = \sum_{H' \in \mathcal{H}_V} p(H,H') \cdot H'
  \]
  is valid.
\end{lem}
\begin{proof}
  We use Lemma 2.2 from~\cite{Razborov2007} in the proof. 

  Let $\phi$ be an arbitrary limiting graph. Then, there exists 
  a converging sequence $(G_n)_{n \in \mathbb{N}}$ of graphs such that 
  $\lim_{n \to \infty} p(F,G_n) = \phi(F)$ for all graphs $F \in \mathcal{G}$. Pick $n_0 \in \mathbb{N}$ such that $v(G_{n_0}) \geq |V|$. Then, since $\mathcal{H}_V$ is a complete set of graphs on $V$, 
  by Lemma 2.2 from~\cite{Razborov2007}, we have
  \[
  p(H,G_n) = \sum_{H' \in \mathcal{H}_V} p(H,H') \cdot p(H',G_n)
  \]
  for all $n \geq n_0$. Therefore,
  \begin{align*}
  \db{H}\phi = \phi(H) & = \lim_{n \to \infty} p(H,G_n)
  \\
  & = \lim_{n \to \infty} \sum_{H' \in \mathcal{H}_V} p(H,H') \cdot p(H',G_n)
  \\
  & = \sum_{H' \in \mathcal{H}_V} p(H,H') \cdot \lim_{n \to \infty} p(H',G_n)
  \\
  & = \sum_{H' \in \mathcal{H}_V} p(H,H') \cdot \phi(H')
  \\
  & = \sum_{H' \in \mathcal{H}_V} p(H,H') \cdot \db{H'}\phi
  = \db{\sum_{H' \in \mathcal{H}_V} p(H,H') \cdot H'}\phi.
  \end{align*}
  Since $\phi$ was arbitrary, the equation
  $H = \sum_{H' \in \mathcal{H}_V} p(H,H') \cdot H'$
  is valid.
\end{proof}
\begin{lem}
  \label{lemma:eliminate-multiplication}
  If $\emptygraph$ is the graph with the empty vertex set, then
  \[
  1 = \emptygraph
  \]
  is valid. Also, for every pair of graphs $H_1$ and $H_2$ with $v(H_1) + v(H_2) \leq |V|$, we have the valid equation
  \[
  H_1 \cdot H_2 = \sum_{H \in \mathcal{H}_V} r_H \cdot H,
  \]
  where
  \[
  r_H \defeq \mathbb{P}\big[H[\mathbf{U}_1] \simeq H_1\ \, \mathrm{and}\ \, H[\mathbf{U}_2] \simeq H_2\big],
  \]
  with $\mathbf{U}_1$ and $\mathbf{U}_2$ being uniformly random disjoint subsets of $V(H)$ of sizes $v(H_1)$ and $v(H_2)$, respectively.
\end{lem}
\begin{proof}
  We use Lemmas 2.2 and 2.3 from~\cite{Razborov2007} in the second part of the proof.

  Let $\phi$ be an arbitrary limiting graph, and let $(G_n)_{n \in \mathbb{N}}$ be a converging sequence of graphs such that 
  $\lim_{n \to \infty} p(F,G_n) = \phi(F)$ for all graphs $F \in \mathcal{G}$. 
  
  Since $p(\emptyset,G_n) = 1$ for all $n \in \mathbb{N}$, we have
  \[
  \db{1}\phi = 1 = \lim_{n \to \infty} p(\emptyset,G_n) = \phi(\emptyset) = \db{\emptygraph}\phi.
  \]
  The validity of $1 = \emptygraph$ follows from the arbitrariness of $\phi$.

  Pick $n_0 \in \mathbb{N}$ such that $v(G_{n_0}) \geq |V|$. Then, for all $n \geq n_0$, applying the aforementioned lemmas from~\cite{Razborov2007} yields
  \[
  \left|
  p(H_1,G_n) \cdot p(H_2,G_n) - 
  \sum_{H \in \mathcal{H}_V} r_H \cdot p(H,G_n)
  \right| 
  \leq
  \frac{(v(H_1) + v(H_2))^C}{v(G_n)}
  \]
  for some constant $C$. Since $v(G_n) \to \infty$ as $n \to \infty$, we have
  \begin{align*}
  \lim_{n \to \infty} (p(H_1,G_n) \cdot p(H_2,G_n)) 
  = 
  \lim_{n \to \infty} \sum_{H \in \mathcal{H}_V} r_H \cdot p(H,G_n).
  \end{align*} 
  Thus,
  \begin{align*}
  \db{H_1 \cdot H_2}\phi 
  = \phi(H_1) \cdot \phi(H_2)
  & {} =
  \lim_{n \to \infty} p(H_1,G_n) \cdot \lim_{n \to \infty} p(H_2,G_n)
  \\
  & {} =
  \lim_{n \to \infty} (p(H_1,G_n) \cdot p(H_2,G_n)) 
  \\
  & {} =
  \lim_{n \to \infty} \sum_{H \in \mathcal{H}_V} r_H \cdot p(H,G_n)
  \\
  & {} =
  \sum_{H \in \mathcal{H}_V} r_H \cdot \lim_{n \to \infty} p(H,G_n)
  = 
  \sum_{H \in \mathcal{H}_V} r_H \cdot \phi(H)
  = \db{\sum_{H \in \mathcal{H}_V} r_H \cdot H}\phi.
  \end{align*}
  Since $\phi$ was arbitrary, the equation
  \[
  H_1 \cdot H_2 = \sum_{H \in \mathcal{H}_V} r_H \cdot H
  \]
  is valid, as claimed.
\end{proof}

We now illustrate the syntax and semantics of flag algebra by expressing the asymptotic extremal problems from \Cref{sec:induced-subgraph-density}.
For a graph $H$ and a finite set of forbidden graphs $\mathcal{F}$ such that there exists at least one graph avoiding every graph in $\mathcal{F}$, the Tur\'an-type problem is to find the smallest $c \in \mathbb{R}$ such that the implication
\begin{equation}
  \label{eqn:turan-assertion}
\left(\bigwedge_{F \in \mathcal{F}} F = 0 \right) \implies H \leq c
\end{equation}
is valid in flag algebra, i.e., it holds for all limiting graphs $\phi$.
Here, $c$ is an abbreviation of the density expression $c \cdot 1$.
The antecedent $\bigwedge_{F \in \mathcal{F}} F = 0$ expresses that $\phi$ avoids all forbidden graphs in $\mathcal{F}$, and the consequent $H \leq c$ states that the density of $H$ in $\phi$ is at most $c$.

For the Ramsey multiplicity problem, given integers $s,t \geq 2$, we seek the largest $c' \in \mathbb{R}$ such that the following inequality is valid in flag algebra:
\begin{equation}
  \label{eqn:ramsey-assertion}
K_s + I_t \geq c'.
\end{equation}
Here, $I_t$ denotes the edgeless graph with $t$ vertices.

How does flag algebra help us to solve the problems of finding such optimal constants $c$ and $c'$ in \Cref{eqn:turan-assertion,eqn:ramsey-assertion}?
The answer lies in reasoning principles for deriving valid assertions in flag algebra. 
We give a glimpse of this aspect with two examples here.

Consider the Tur\'an-type problem with $H = K_2$ and $\mathcal{F} = \{K_3\}$, which asks for the best upper bound on the asymptotic edge density of triangle-free graphs.
In the language of flag algebra, this asks for the smallest $c$ such that the implication
\[
K_3 = 0 \implies K_2 \leq c
\]
is valid.
By Mantel's theorem~\cite{Mantel1907}, the optimal constant is $c = 1/2$.
We prove validity for $c = 1/2$.

Write $\Ithree$, $\Ethree$, $\Pthree$, and $\Kthree$ for the edgeless graph, the single-edge graph, the two-edge graph, and the complete graph $K_3$ on the vertex set $[3]$, respectively.\footnote{The exact correspondence between the dots in these notations and the labels $1,2,3$ is immaterial; any fixed choice yields isomorphic graphs.}
We use the following two valid assertions:
\begin{align}
  & K_2 = \Big(0 \cdot \Ithree + \tfrac{1}{3} \cdot \Ethree + \tfrac{2}{3} \cdot \Pthree + 1 \cdot \Kthree\Big),
  \label{eqn:mantel-density-decomposition}
  \\[2ex]
  & \Big(\Ithree - \tfrac{1}{3} \cdot \Ethree - \tfrac{1}{3} \cdot \Pthree + \Kthree\Big) \,\geq\, 0.
  \label{eqn:mantel-key-inequality}
\end{align}
The first is an instance of \Cref{lemma:decomposition}; the second can be proved using the reasoning principle explained in the next section.
Assuming $K_3 = 0$, we derive the desired bound as follows (each step holds for all $\phi \in \Phi$ with $\phi(K_3)=0$):
\begin{align*}
  K_2
  & = 0 \cdot \Ithree + \tfrac{1}{3} \cdot \Ethree + \tfrac{2}{3} \cdot \Pthree + 1 \cdot \Kthree
  && \text{by \Cref{eqn:mantel-density-decomposition}}
  \\
  & \leq
  \begin{aligned}[t]
  & \left(0 \cdot \Ithree + \tfrac{1}{3} \cdot \Ethree + \tfrac{2}{3} \cdot \Pthree + 1 \cdot \Kthree \right)
  \\
  & {} + \tfrac{1}{2} \cdot
  \left(
  \Ithree - \tfrac{1}{3} \cdot \Ethree - \tfrac{1}{3} \cdot \Pthree + \Kthree\right)
  \end{aligned}
  && \text{by \Cref{eqn:mantel-key-inequality,lemma:basic-algebra}}
  \\
  & = \tfrac{1}{2} \cdot \Ithree + \tfrac{1}{6} \cdot \Ethree + \tfrac{1}{2} \cdot \Pthree + \tfrac{3}{2} \cdot \Kthree
  && \text{by \Cref{lemma:basic-algebra}}
  \\
  & = \tfrac{1}{2} \cdot \Ithree + \tfrac{1}{6} \cdot \Ethree + \tfrac{1}{2} \cdot \Pthree + \tfrac{1}{2} \cdot \Kthree
  && \text{since } \Kthree = K_3 = 0
  \\
  & \leq \tfrac{1}{2} \cdot \Ithree + \tfrac{1}{2} \cdot \Ethree + \tfrac{1}{2} \cdot \Pthree + \tfrac{1}{2} \cdot \Kthree
  && \text{by \Cref{lemma:non-negativity,lemma:basic-algebra}}
  \\
  & = \tfrac{1}{2} \cdot \left(\Ithree + \Ethree + \Pthree + \Kthree\right)
  && \text{by \Cref{lemma:basic-algebra}}
  \\
  & = \tfrac{1}{2} \cdot 1
  && \text{by \Cref{lemma:sum-to-one}}
  \\
  & = \tfrac{1}{2}.
\end{align*}

As a second example, consider the Ramsey multiplicity problem with $s=t=3$, which asks for the largest $c'$ such that
\[
I_3 + K_3 \geq c'
\]
is valid in flag algebra.
By Goodman's result~\cite{Goodman1959}, the optimal constant is $c' = 1/4$.
We prove validity for $c' = 1/4$ as follows:
\begin{align*}
  I_3 + K_3
  & = 1 \cdot \Ithree + 0 \cdot \Ethree + 0 \cdot \Pthree + 1 \cdot \Kthree
  \\
  & \geq
  \begin{aligned}[t]
  & \left(1 \cdot \Ithree + 0 \cdot \Ethree + 0 \cdot \Pthree + 1 \cdot \Kthree\right)
  \\
  & {} -
  \tfrac{3}{4} \cdot
  \left(
  \Ithree - \tfrac{1}{3} \cdot \Ethree - \tfrac{1}{3} \cdot \Pthree + \Kthree
  \right)
  \end{aligned}
  && \text{by \Cref{eqn:mantel-key-inequality,lemma:basic-algebra}}
  \\
  & = \tfrac{1}{4} \cdot \Ithree + \tfrac{1}{4} \cdot \Ethree + \tfrac{1}{4} \cdot \Pthree + \tfrac{1}{4} \cdot \Kthree
  && \text{by \Cref{lemma:basic-algebra}}
  \\
  & = \tfrac{1}{4} \cdot (\Ithree + \Ethree + \Pthree + \Kthree)
  && \text{by \Cref{lemma:basic-algebra}}
  \\
  & = \tfrac{1}{4} \cdot 1
  && \text{by \Cref{lemma:sum-to-one}}
  \\
  & = \tfrac{1}{4}.
\end{align*}

Both examples use a step that has not yet been justified, namely the inequality in \Cref{eqn:mantel-key-inequality}.
The next section explains a reasoning principle of flag algebra that allows one to derive such inequalities, as well as the downward operator, a key tool for applying this principle.

\section{Deriving Inequalities with the Downward Operator}
\label{sec:downward-operator}

In this section, we explain a common proof strategy in flag algebra. It allows one to prove nontrivial inequalities such as \Cref{eqn:mantel-key-inequality}. The main tool is the \emph{downward operator}, which transfers inequalities proved in a labelled setting back to the original (unlabelled) setting.

At a high level, the strategy proceeds by introducing a family of \emph{labelled} variants of flag algebra and, for each variant, a sound mechanism for transferring valid inequalities in that variant to valid inequalities in the original flag algebra. Each variant is chosen so that certain inequalities become easier to prove. One then proves such ``easy'' inequalities in an appropriate variant and transfers them to the unlabelled setting. In the derivation of \Cref{eqn:mantel-key-inequality}, for example, we use the variant in which one vertex is labelled.

To make this precise, we need some additional definitions.
A \emph{type} $\tau$ is a graph on the vertex set $[k]$ for some $k \in \mathbb{N}$; we call $k$ the \emph{size} of $\tau$.
A \emph{$\tau$-labelled graph} is a pair $G^\tau = (G, \theta)$ consisting of a graph $G$ and an injective map $\theta : [k] \to V(G)$ such that $\theta$ is an \emph{embedding} of $\tau$ into $G$, i.e., for all $i,j \in [k]$,
\[
\{i,j\} \in E(\tau) \quad\text{if and only if}\quad \{\theta(i), \theta(j)\} \in E(G).
\]
We call $\tau$ the \emph{type} of $G^\tau$, and write $\lvert G^\tau \rvert$ for the underlying unlabelled graph $G$.
We also write $V(G^\tau)$, $E(G^\tau)$, $v(G^\tau)$, and $e(G^\tau)$ for $V(G)$, $E(G)$, $v(G)$, and $e(G)$, respectively.
For a fixed type $\tau$, we denote the set of all $\tau$-labelled graphs by $\mathcal{G}^\tau$.
Note that $\mathcal{G}^\tau$ is countably infinite, since every graph considered in this survey article has a finite vertex set contained in $\mathbb{N}$.

We now repeat the concepts from the previous sections for $\tau$-labelled graphs.
For $\tau$-labelled graphs $H^\tau = (H,\iota)$ and $G^\tau = (G,\theta)$, a function $h : V(H) \to V(G)$ is a \emph{$\tau$-homomorphism} if it is a homomorphism from $H$ to $G$ and satisfies $h(\iota(i)) = \theta(i)$ for all $i \in [k]$.
We say that $H^\tau$ and $G^\tau$ are \emph{$\tau$-isomorphic}, denoted $H^\tau \simeq_\tau G^\tau$, if there exists a bijective $\tau$-homomorphism from $H^\tau$ to $G^\tau$ whose inverse is also a $\tau$-homomorphism from $G^\tau$ to $H^\tau$.

For $\tau$-labelled graphs $H^\tau = (H,\iota)$ and $G^\tau = (G,\theta)$, the \emph{induced $\tau$-labelled subgraph density} of $H^\tau$ in $G^\tau$ (or simply the \emph{$H^\tau$-density} in $G^\tau$) is defined by
\[
p_\tau(H^\tau,G^\tau) \defeq \mathbb{P}\Big[\big(G\big[\theta([k]) \cup \mathbf{U}\big],\,\theta\big) \simeq_\tau H^\tau\Big],
\]
where $\mathbf{U}$ is a uniformly random subset of $V(G) \setminus \theta([k])$ of size $v(H^\tau) - k$.
A sequence of $\tau$-labelled graphs $(G^\tau_n)_{n \in \mathbb{N}}$ is \emph{increasing} if $v(G^\tau_n) < v(G^\tau_{n+1})$.
It \emph{converges} if it is increasing and $\lim_{n \to \infty} p_\tau(H^\tau, G^\tau_n)$ exists for every $\tau$-labelled graph $H^\tau \in \mathcal{G}^\tau$.
A function $\phi : \mathcal{G}^\tau \to [0,1]$ is a \emph{$\tau$-labelled limiting graph} if there exists a convergent sequence $(G^\tau_n)_{n \in \mathbb{N}}$ such that
$\phi(H^\tau) = \lim_{n \to \infty} p_\tau(H^\tau, G^\tau_n)$ for every $H^\tau \in \mathcal{G}^\tau$.
As in the unlabelled case, $\phi$ can be viewed (informally) as a random infinite graph; however, now $k$ distinguished vertices in this infinite graph are labelled $1,\ldots,k$, and these labels determine an embedding of $\tau$ into the infinite graph.
We write $\Phi^\tau$ for the set of all $\tau$-labelled limiting graphs.

The $\tau$-labelled variant of flag algebra is defined analogously to the original one, replacing graphs, induced subgraph densities, convergent graph sequences, and limiting graphs by $\tau$-labelled graphs, induced $\tau$-labelled subgraph densities, convergent sequences of $\tau$-labelled graphs, and $\tau$-labelled limiting graphs, respectively.
Its syntax consists of density expressions and assertions built from $\tau$-labelled graphs and interpreted over $\Phi^\tau$.
For convenience, we recall the syntax:
\begin{align*}
\text{Density Expressions}\quad E^\tau & ::= H^\tau \,\mid\, r \cdot E^\tau \,\mid\, 0^\tau \,\mid\, E^\tau + E^\tau \,\mid\, 1^\tau \,\mid\, E^\tau \cdot E^\tau,
\\
\text{Assertions}\quad A^\tau & ::= \lfalse \,\mid\, \ltrue \,\mid\, E^\tau \geq E^\tau \,\mid\, \neg A^\tau \,\mid\, A^\tau \lor A^\tau,
\end{align*}
where $H^\tau$ ranges over $\tau$-labelled graphs and $r$ ranges over real numbers.
We write $\Expr^\tau$ for the set of density expressions and $\Assn^\tau$ for the set of assertions in this variant.

All of \Cref{lemma:basic-algebra,lemma:isomorphism,lemma:non-negativity,lemma:sum-to-one,lemma:decomposition,lemma:eliminate-multiplication} also hold in the $\tau$-labelled variant, with the obvious modifications (replacing graphs by $\tau$-labelled graphs and induced subgraph densities by induced $\tau$-labelled subgraph densities).
We record one immediate consequence of \Cref{lemma:basic-algebra} and the labelled version of \Cref{lemma:eliminate-multiplication} for later use.
\begin{cor}
  \label{cor:basic-algebra-tau:sos}
  Let $\tau$ be a type of size $k$. Then, for all expressions $E_1^\tau,\ldots,E_m^\tau$ in $\Expr^\tau$,
  \[
    \Big(\sum_{i \in [m]} E_i^\tau \cdot E_i^\tau\Big) \geq 0^\tau
  \]
  is a valid assertion in the $\tau$-labelled variant of flag algebra.
\end{cor}
\begin{lem}
  \label{lemma:eliminate-multiplication-tau}
  Let $\tau$ be a type of size $k$. Then the equation
  \[
    1^\tau = (\tau,\mathrm{id}_{[k]})
  \]
  is valid in the $\tau$-labelled variant of flag algebra. Also, let $V \subseteq \mathbb{N}$ be finite with $|V| \geq k$, and let $\mathcal{H}^\tau_V$ be a set of $\tau$-labelled graphs on $V$ such that every $\tau$-labelled graph on $V$ is $\tau$-isomorphic to exactly one element of $\mathcal{H}^\tau_V$. Then, for every pair of $\tau$-labelled graphs $H_1^\tau$ and $H_2^\tau$ with $v(H_1^\tau) + v(H_2^\tau) - k \leq |V|$, the following equation is valid:
  \[
    H_1^\tau \cdot H_2^\tau = \sum_{(H,\theta) \in \mathcal{H}^\tau_V} r_{(H,\theta)} \cdot (H,\theta),
  \]
  where
  \[
    r_{(H,\theta)} \defeq \mathbb{P}\big[H\big[\theta([k]) \cup \mathbf{U}_1\big] \simeq_\tau H_1^\tau\ \, \mathrm{and}\ \, H\big[\theta([k]) \cup \mathbf{U}_2\big] \simeq_\tau H_2^\tau\big],
  \]
  with $\mathbf{U}_1$ and $\mathbf{U}_2$ being uniformly random disjoint subsets of $V(H) \setminus \theta([k])$ of sizes $v(H_1^\tau) - k$ and $v(H_2^\tau) - k$, respectively.
\end{lem}

The transfer mechanism from the $\tau$-labelled variant to the original flag algebra is the \emph{downward operator}. We define it by a general recipe that relates the variant and the original flag algebras via what we call an \emph{adjoint pair of functions}. We explain this recipe next and then instantiate it to obtain the downward operator.

Fix a type $\tau$, and consider functions $\gamma : \mathcal{G} \to \FinMea(\mathcal{G}^\tau)$ and $\alpha : \mathcal{G}^\tau \to \FinMea(\mathcal{G})$.
Here we regard $\mathcal{G}$ and $\mathcal{G}^\tau$ as measurable spaces equipped with the discrete $\sigma$-algebras, and write $\FinMea(\mathcal{G})$ and $\FinMea(\mathcal{G}^\tau)$ for the sets of finite measures on $\mathcal{G}$ and $\mathcal{G}^\tau$, respectively.
Intuitively, for $G \in \mathcal{G}$, the measure $\gamma(G)$ defines an unnormalised probability distribution over $\mathcal{G}^\tau$, and for $H^\tau \in \mathcal{G}^\tau$, the measure $\alpha(H^\tau)$ defines an unnormalised probability distribution over $\mathcal{G}$.
We say that $(\alpha,\gamma)$ is an \emph{adjoint pair} if, for all $H^\tau \in \mathcal{G}^\tau$ and $G \in \mathcal{G}$,
\begin{equation}
\label{eq:adjoint-pair}
p(\alpha(H^\tau), G) = p_\tau(H^\tau, \gamma(G)),
\end{equation}
where
\begin{align*}
  p(\alpha(H^\tau), G) &\defeq \int_{\mathcal{G}} p(H, G) \, \Big(\alpha(H^\tau)(dH)\Big),
  \\
  p_\tau(H^\tau, \gamma(G)) &\defeq \int_{\mathcal{G}^\tau} p_\tau(H^\tau, G^\tau) \, \Big(\gamma(G)(dG^\tau)\Big).
\end{align*}

Equation~\eqref{eq:adjoint-pair} is reminiscent of a Galois connection or a categorical adjunction.\footnote{Readers familiar with functional analysis may also recognise the similarity with adjoint linear operators.}
In automated verification, Galois connections are used to relate different interpretations of a programming language or a logical system, and in programming languages, categorical adjunctions are used to relate semantic settings.
Similarly, we use adjoint pairs to relate the original flag algebra and its $\tau$-labelled variant.

Let $(\alpha,\gamma)$ be an adjoint pair for a type $\tau$.
We say that $\alpha$ has \emph{finite support} if, for every $H^\tau \in \mathcal{G}^\tau$, the support
\[
\support(\alpha(H^\tau)) \defeq \{ H \,:\ \alpha(H^\tau)(\{H\}) > 0 \}
\]
is finite. This condition allows us to define the following map $\alpha^\dagger$ from multiplication-free and $1^\tau$-free expressions in $\Expr^\tau$ to density expressions in the original flag algebra:
\begin{align*}
  \alpha^\dagger(H^\tau) & \defeq \sum_{H \in \support(\alpha(H^\tau))} \alpha(H^\tau)(\{H\}) \cdot H,
  &
  \alpha^\dagger(r \cdot E^\tau) & \defeq r \cdot \alpha^\dagger(E^\tau),
  \\
  \alpha^\dagger(0^\tau) & \defeq 0,
  &
  \alpha^\dagger(E_1^\tau + E_2^\tau) & \defeq \alpha^\dagger(E_1^\tau) + \alpha^\dagger(E_2^\tau).
\end{align*}
We also say that $\gamma$ is \emph{extendible to limiting graphs} if there exists a function $\gamma_\infty : \Phi \to \FinMea(\Phi^\tau)$ such that, for every convergent sequence of graphs $(G_n)_{n \in \mathbb{N}}$ and every limiting graph $\phi \in \Phi$,
\[
\Bigl(\forall H \in \mathcal{G},\, \lim_{n \to \infty} p(H, G_n) = \phi(H)\Bigr)
\implies
\Bigl(\forall H^\tau \in \mathcal{G}^\tau,\, \lim_{n \to \infty} p_\tau(H^\tau, \gamma(G_n)) = \int_{\Phi^\tau} \phi^\tau(H^\tau) \,\bigl(\gamma_\infty(\phi)(d\phi^\tau)\bigr)\Bigr).
\]
Here, $\Phi^\tau$ is equipped with the $\sigma$-algebra induced by the product topology on $[0,1]^{\mathcal{G}^\tau}$, making it a measurable space, and $\gamma_\infty(\phi)$ is a finite measure on $\Phi^\tau$.

The next lemma states that if $\alpha$ has finite support and $\gamma$ is extendible to limiting graphs, then the adjoint pair $(\alpha,\gamma)$ yields a sound transfer principle.
\begin{lem}
  \label{lemma:adjoint-transfer}
  Let $(\alpha,\gamma)$ be an adjoint pair for a type $\tau$.
  Assume that $\alpha$ has finite support, and let $\alpha^\dagger$ be the induced extension of $\alpha$ to multiplication-free expressions that do not contain $1^\tau$.
  Also assume that $\gamma$ is extendible to limiting graphs.
  Then, for every valid inequality $E^\tau \geq 0^\tau$ in the $\tau$-labelled variant, where $E^\tau$ is multiplication-free and $1^\tau$-free, the inequality $\alpha^\dagger(E^\tau) \geq 0$ is valid in the original flag algebra.
\end{lem}
\begin{proof}
  Let $\gamma_\infty : \Phi \to \FinMea(\Phi^\tau)$ be a witness of the extendibility 
  of $\gamma$ to limiting graphs.
  Consider an arbitrary limiting graph $\phi \in \Phi$. We need to show that $\db{\alpha^\dagger(E^\tau)}\phi \geq 0$. 
  Using structural induction on $E^\tau$, we will show that 
  \begin{equation}
    \label{eq:lem:adjoint-transfer0}
    \db{\alpha^\dagger(E^\tau)}\phi = \int_{\Phi^\tau} \db{E^\tau}\phi^\tau \,\Big(\gamma_\infty(\phi)(d\phi^\tau)\Big).
  \end{equation}
  The desired conclusion then follows from this equation because $E^\tau \geq 0^\tau$ is valid in the $\tau$-labelled variant
  and so the integrand $\db{E^\tau}\phi^\tau$ is non-negative for every $\phi^\tau \in \Phi^\tau$.

  Assume that $E^\tau$ is a $\tau$-labelled graph $H^\tau$. By the definition of a limiting graph, there exists a convergent sequence of graphs $(G_n)_n$ such that 
  \begin{equation}
    \label{eq:lem:adjoint-transfer1}
    \phi(H) = \lim_{n \to \infty} p(H, G_n)  \quad\text{for all } H \in \mathcal{G}.
  \end{equation}
  Using this, we derive \Cref{eq:lem:adjoint-transfer0} as follows:
  \begin{align*}
    \db{\alpha^\dagger(H^\tau)}\phi 
    & {} = \db{\sum_{H \in \support(\alpha(H^\tau))} \alpha(H^\tau)(\{H\}) \cdot H}\phi
    && \text{by the definition of } \alpha^\dagger
    \\
    & {} = \sum_{H \in \support(\alpha(H^\tau))} \alpha(H^\tau)(\{H\}) \cdot \phi(H)
    && \text{by the definition of } \db{-}
    \\
    & {} = \sum_{H \in \support(\alpha(H^\tau))} \alpha(H^\tau)(\{H\}) \cdot \lim_{n \to \infty} p(H, G_n)
    && \text{by \Cref{eq:lem:adjoint-transfer1}}
    \\
    & {} = \lim_{n \to \infty} \sum_{H \in \support(\alpha(H^\tau))} \alpha(H^\tau)(\{H\}) \cdot p(H, G_n)
    && \text{by the finiteness of } \support(\alpha(H^\tau))
    \\
    & {} = \lim_{n \to \infty} p(\alpha(H^\tau), G_n)
    && \text{by the definition of } p(\alpha(H^\tau), G_n)
    \\
    & {} = \lim_{n \to \infty} p_\tau(H^\tau, \gamma(G_n))
    && \text{by the adjointness of } (\alpha,\gamma)
    \\
    & {} = \int_{\Phi^\tau} \phi^\tau(H^\tau) \, \Big(\gamma_\infty(\phi)(d\phi^\tau)\Big)
    && \text{by the extendibility of } \gamma
    \\
    & {} = \int_{\Phi^\tau} \db{H^\tau}\phi^\tau \, \Big(\gamma_\infty(\phi)(d\phi^\tau)\Big)
    && \text{by the definition of } \db{-}.
  \end{align*}

  Next consider the case that $E^\tau = 0^\tau$. Then, 
  \begin{align*}
    \db{\alpha^\dagger(0^\tau)}\phi 
    = \db{0}\phi
    = 0 
    = \int_{\Phi^\tau} 0 \, \Big(\gamma_\infty(\phi)(d\phi^\tau)\Big)
    = \int_{\Phi^\tau} \db{0^\tau}\phi^\tau \, \Big(\gamma_\infty(\phi)(d\phi^\tau)\Big).
  \end{align*}

  The remaining cases can be proved by the induction hypothesis and the linearity of integration. We first derive the target equation for $E^\tau = r \cdot F^\tau$:
  \begin{align*}
    \db{\alpha^\dagger(r \cdot F^\tau)}\phi 
    & = \db{r \cdot \alpha^\dagger(F^\tau)}\phi
    && \text{by the definition of } \alpha^\dagger
    \\
    & = r \cdot \left(\db{\alpha^\dagger(F^\tau)}\phi\right)
    && \text{by the definition of } \db{-}
    \\
    & = r \cdot \int_{\Phi^\tau} \db{F^\tau}\phi^\tau \, \Big(\gamma_\infty(\phi)(d\phi^\tau)\Big)
    && \text{by the induction hypothesis}
    \\
    & = \int_{\Phi^\tau} r \cdot \left(\db{F^\tau}\phi^\tau\right) \, \Big(\gamma_\infty(\phi)(d\phi^\tau)\Big)
    && \text{by the linearity of integration}
    \\
    & = \int_{\Phi^\tau} \db{r \cdot F^\tau}\phi^\tau \, \Big(\gamma_\infty(\phi)(d\phi^\tau)\Big)
    && \text{by the definition of } \db{-}.
  \end{align*}
  Next, we handle the case that $E^\tau = E_1^\tau + E_2^\tau$:
  \begin{align*}
    & \db{\alpha^\dagger(E_1^\tau + E_2^\tau)}\phi 
    \\
    & \quad {} = \db{\alpha^\dagger(E_1^\tau) + \alpha^\dagger(E_2^\tau)}\phi
    && \text{by the definition of } \alpha^\dagger
    \\
    & \quad {} = \db{\alpha^\dagger(E_1^\tau)}\phi + \db{\alpha^\dagger(E_2^\tau)}\phi
    && \text{by the definition of } \db{-}
    \\
    & \quad {} = \int_{\Phi^\tau} \db{E_1^\tau}\phi^\tau \, \Big(\gamma_\infty(\phi)(d\phi^\tau)\Big)
    + \int_{\Phi^\tau} \db{E_2^\tau}\phi^\tau \, \Big(\gamma_\infty(\phi)(d\phi^\tau)\Big)
    && \text{by the induction hypothesis}
    \\
    & \quad {} = \int_{\Phi^\tau} \left(\db{E_1^\tau}\phi^\tau + \db{E_2^\tau}\phi^\tau\right) \, \Big(\gamma_\infty(\phi)(d\phi^\tau)\Big)
    && \text{by the linearity of integration}
    \\
    & \quad {} = \int_{\Phi^\tau} \db{E_1^\tau + E_2^\tau}\phi^\tau \, \Big(\gamma_\infty(\phi)(d\phi^\tau)\Big)
    && \text{by the definition of } \db{-}.     
  \end{align*}
\end{proof}

We are ready to define the downward operator for each type $\tau$. Given an expression $E^\tau$ in $\Expr^\tau$, we first eliminate all multiplications and occurrences of the constant $1^\tau$ in $E^\tau$ using \Cref{lemma:eliminate-multiplication-tau}, obtaining a multiplication-free, $1$-free expression $F^\tau$ that is equal to $E^\tau$ in the $\tau$-labelled variant of flag algebra. Although this elimination step is underspecified, it is typically performed from the innermost subexpression outward. The downward operator then maps $F^\tau$ to $\alpha_d^\dagger(F^\tau)$, where $\alpha_d$ is part of a specific adjoint pair $(\alpha_d,\gamma_d)$ for the type $\tau$ defined below. This adjoint pair satisfies the conditions in \Cref{lemma:adjoint-transfer}. Consequently, if $E^\tau \geq 0^\tau$ is a valid inequality in the $\tau$-labelled variant, then $\alpha_d^\dagger(F^\tau) \geq 0$ is valid in the original flag algebra.

To define the adjoint pair $(\alpha_d,\gamma_d)$ mentioned above, consider a $\tau$-labelled graph $H^\tau \in \mathcal{G}^\tau$ and an unlabelled graph $G \in \mathcal{G}$. For each $G^\tau \in \mathcal{G}^\tau$, let $\boldsymbol{\theta}_{G^\tau}$ be the uniformly random injection from $[k]$ to $V(G^\tau)$, and also let
\[
  q_{G^\tau} \defeq \mathbb{P}\Big[\boldsymbol{\theta}_{G^\tau} \text{ is an embedding from $\tau$ to $|G^\tau|$ and $(|G^\tau|,\boldsymbol{\theta}_{G^\tau}) \simeq_\tau G^\tau$}\Big].
\]
Note that $(\tau,\mathrm{id}_{[k]})$ is a $\tau$-labelled graph whose underlying graph is $\tau$ itself, 
and that $q_{(\tau,\mathrm{id}_{[k]})} = |\mathrm{Aut}(\tau)| / v(\tau)!$ where $\mathrm{Aut}(\tau)$ is the set of all automorphisms on $\tau$. Using this notation,
we define finite measures $\alpha_d(H^\tau) \in \FinMea(\mathcal{G})$ and $\gamma_d(G) \in \FinMea(\mathcal{G}^\tau)$ as follows:
\begin{align*}
  \alpha_d(H^\tau)(\{H\}) & \defeq 
  \begin{cases}
  q_{H^\tau} & \text{if } |H^\tau| = H
  \\
  0 & \text{otherwise},
  \end{cases}
  \\
  \gamma_d(G)(\{G^\tau\}) & \defeq 
  \begin{cases}
  \big(q_{(\tau,\mathrm{id})} \cdot p(\tau,G)\big) \div \big|\{G^\tau_0 \in \mathcal{G}^\tau \,:\ |G^\tau_0| = G\}\big| & \text{if } |G^\tau| = G
  \\
  0 & \text{otherwise}.
  \end{cases}
\end{align*}
Among the conditions in \Cref{lemma:adjoint-transfer}, it is straightforward to see that $\alpha_d(H^\tau)$ has finite support for every $H^\tau \in \mathcal{G}^\tau$. Also, it can be shown that $\gamma_d$ is extendible to limiting graphs, although the proof is involved. See the proof of Theorem 3.5 in \cite{Razborov2007} for details. Finally, the following lemma establishes the adjointness of $(\alpha_d,\gamma_d)$.
\begin{lem}
  \label{lemma:downward-adjoint}
  The pair $(\alpha_d,\gamma_d)$ defined above is an adjoint pair for the type $\tau$.
\end{lem}
\begin{proof}
  We prove the lemma by substantially expanding the argument in the proof of Lemma 3.11 in \cite{Razborov2007} and adjusting the expanded version to our setting.

  Fix an arbitrary $\tau$-labelled graph $H^\tau \in \mathcal{G}^\tau$ and an arbitrary unlabelled
  graph $G \in \mathcal{G}$. We prove that
  \[
    p(\alpha_d(H^\tau), G) = p_\tau(H^\tau, \gamma_d(G)).
  \]

  We first simplify the left-hand side:
  \begin{align*}
    p(\alpha_d(H^\tau), G)
    & = \int_{\mathcal{G}} p(H, G) \, \big(\alpha_d(H^\tau)(dH)\big)
    = q_{H^\tau} \cdot p(|H^\tau|, G)
    && \text{by the definition of } \alpha_d(H^\tau).
  \end{align*}

  Let
  \[
    Z_G \defeq \big|\{G^\tau \in \mathcal{G}^\tau \,:\ |G^\tau| = G\}\big|.
  \]
  If $Z_G = 0$, then $p(\tau,G)=0$ (there is no embedding of $\tau$ into $G$), hence $\gamma_d(G)$ is
  the zero measure and therefore $p_\tau(H^\tau,\gamma_d(G))=0$.
  Moreover, since $H^\tau$ contains $\tau$ on the labelled vertices, an induced copy of $|H^\tau|$ in
  $G$ would in particular contain an induced copy of $\tau$, contradicting $p(\tau,G)=0$; hence
  $p(|H^\tau|,G)=0$ and so the left-hand side is also $0$.
  Thus, we may assume $Z_G>0$.

  Under this assumption, we simplify the right-hand side:
  \begin{align*}
    p_\tau(H^\tau, \gamma_d(G))
    & = \int_{\mathcal{G}^\tau} p_\tau(H^\tau, G^\tau) \, \big(\gamma_d(G)(dG^\tau)\big)
    \\
    & = \sum_{\substack{G^\tau \in \mathcal{G}^\tau \\ |G^\tau| = G}}
    p_\tau(H^\tau, G^\tau)\cdot \frac{q_{(\tau,\mathrm{id})}\cdot p(\tau,G)}{Z_G}
    && \text{by the definition of } \gamma_d(G)
    \\
    & = \big(q_{(\tau,\mathrm{id})}\cdot p(\tau,G)\big)\cdot
    \sum_{\substack{G^\tau \in \mathcal{G}^\tau \\ |G^\tau| = G}}
    \frac{1}{Z_G}\,p_\tau(H^\tau, G^\tau).
  \end{align*}

  Therefore, it suffices to show that
  \begin{equation}
    \label{eq:downward-adjoint:target-eq}
    q_{H^\tau} \cdot p(|H^\tau|, G)
    =
    \big(q_{(\tau,\mathrm{id})} \cdot p(\tau,G)\big)\cdot
    \sum_{\substack{G^\tau \in \mathcal{G}^\tau \\ |G^\tau| = G}}
    \frac{1}{Z_G}\,p_\tau(H^\tau, G^\tau).
  \end{equation}

  Let $k = v(\tau)$, $m = v(H^\tau)$, and $n = v(G)$. Define $\boldsymbol{\theta}$ as a uniformly random injection from $[k]$ to $V(H^\tau)$, and  $\boldsymbol{\theta}'$ as a uniformly random injection from $[k]$ to $V(G)$. Also, let $\mathbf{U}$ be a uniformly random subset of $V(G)$ of size $m$, and $\mathbf{U}'$ be a uniformly random subset of $V(G) \setminus \boldsymbol{\theta}'([k])$ of size $m-k$. Using these random variables, we define two events $E$ and $E'$ as follows:
    \begin{itemize}
      \item $E$ is the event that $\boldsymbol{\theta}$ is an embedding of $\tau$ into $|H^\tau|$, the $\tau$-labelled graph $(|H^\tau|, \boldsymbol{\theta})$ is $\tau$-isomorphic to $H^\tau$, and the induced subgraph $G[\mathbf{U}]$ is isomorphic to $|H^\tau|$.
      \item $E'$ is the event that $\boldsymbol{\theta}'$ is an embedding of $\tau$ into $G$, and the $\tau$-labelled graph $(G[\mathbf{U}' \cup \boldsymbol{\theta}'([k])],\boldsymbol{\theta}')$ is $\tau$-isomorphic to $H^\tau$.
    \end{itemize}

  The left-hand side of \Cref{eq:downward-adjoint:target-eq} is equal to the probability of $E$. The random variables $\boldsymbol{\theta}$ and $\mathbf{U}$ are independent, and $q_{H^\tau}$ and $p(|H^\tau|,G)$ are the probabilities that $\boldsymbol{\theta}$ and $\mathbf{U}$ satisfy their respective conditions in the event $E$. Similarly, the right-hand side of \Cref{eq:downward-adjoint:target-eq} is equal to the probability of $E'$. The factor $q_{(\tau ,\mathrm{id})} \cdot p(\tau,G)$ in the right-hand side is the probability of the event $E''$ that $\boldsymbol{\theta}'$ is an embedding of $\tau$ into $G$, and the remaining factor equals the conditional probability $\mathbb{P}[E' \mid E'']$.  
  
  Thus, it suffices to show that $\mathbb{P}[E] = \mathbb{P}[E']$. We note that the joint distribution of $(\boldsymbol{\theta}, \mathbf{U})$ is the uniform distribution over the set of all pairs $(\theta, U)$ where $\theta$ is an injection from $[k]$ to $V(H^\tau)$ and $U$ is a subset of $V(G)$ with size $m$. Similarly, the joint distribution of $(\boldsymbol{\theta}', \mathbf{U}')$ is the uniform distribution over the set of all pairs $(\theta', U')$ where $\theta'$ is an injection from $[k]$ to $V(G)$ and $U'$ is a subset of $V(G) \setminus \theta'([k])$ with size $m-k$. Elementary counting confirms that these two sample spaces have the same cardinality $n! / ((n-m)! (m-k)!)$. As a result, for such $(U,\theta)$ and $(U',\theta')$, we have
  \[
  \mathbb{P}\big[(\boldsymbol{\theta},\mathbf{U}) = (\theta, U)\big] = \mathbb{P}\big[(\boldsymbol{\theta}',\mathbf{U}') = (\theta', U')\big] = \frac{(n-m)! (m-k)!}{n!}.
  \]
  Therefore, to show that $\mathbb{P}[E] = \mathbb{P}[E']$, we just need to construct a bijection between the set of those $(\theta, U)$ satisfying $E$ and the set of those $(\theta', U')$ satisfying $E'$. For each 
  $U_0 \subseteq V(G)$ with size $m$ such that $G[U_0]$ is isomorphic to $|H^\tau|$, we fix one isomorphism $f_{U_0} : V(H^\tau) \to U_0$ from $|H^\tau|$ to $G[U_0]$. Using this, we construct the desired bijection.
  Given a pair $(\theta, U)$ satisfying $E$, we define the corresponding pair $(\theta', U')$ satisfying $E'$ as follows. By the definition of $E$, $G[U]$ is isomorphic to $|H^\tau|$. So, we have the isomorphism $f_U$ 
  from $|H^\tau|$ to $G[U]$. We define $\theta' : [k] \to V(G)$ by $\theta' \defeq f_U \circ \theta$, and also define $U' \defeq U \setminus \theta'([k])$. Then, $(\theta', U')$ satisfies the event $E'$. Furthermore, this construction is invertible: given a pair $(\theta', U')$ satisfying $E'$, we can recover the corresponding pair $(\theta, U)$ satisfying $E$ by letting $U \defeq U' \cup \theta'([k])$ and $\theta \defeq f_U^{-1} \circ \theta'$. This completes the proof of \Cref{eq:downward-adjoint:target-eq} and thus of the lemma.
\end{proof}

We illustrate the downward operator by using it to derive \Cref{eqn:mantel-key-inequality}.
Let $\tau$ be the unique type of size $1$, i.e., the graph with a single vertex and no edges.
A $\tau$-labelled graph $(G,\theta)$ can be equivalently viewed as a graph $G$ with one distinguished
vertex (namely $\theta(1)$); we adopt this viewpoint below.
By \Cref{cor:basic-algebra-tau:sos}, the inequality
\[
\left(\LItwo - \LEtwo\right)\cdot\left(\LItwo - \LEtwo\right) \geq 0^\tau
\]
is valid in the $\tau$-labelled variant of flag algebra.
Here, $\LItwo$ and $\LEtwo$ are the $\tau$-labelled edgeless graph and the $\tau$-labelled single-edge
graph on vertex set $[2]$, respectively; the distinguished vertex is depicted as an empty circle.

We expand the product and then eliminate all multiplications using valid equalities:
\begin{align*}
  & \left(\LItwo - \LEtwo\right)\cdot\left(\LItwo - \LEtwo\right)
  \\
  & \qquad {} = \LItwo \cdot \LItwo - 2 \cdot \left(\LItwo \cdot \LEtwo\right) + \LEtwo \cdot \LEtwo
  && \text{by the $\tau$-labelled version of \Cref{lemma:basic-algebra}}
  \\
  & \qquad {} = \left(\LIthree + \LEthreeC\right)
  - 2 \cdot \left(\tfrac{1}{2} \cdot \LEthreeB + \tfrac{1}{2} \cdot \LPthreeB\right)
  + \left(\LPthreeC + \LKthree\right)
  && \text{by \Cref{lemma:eliminate-multiplication-tau}}
  \\
  & \qquad {} = \LIthree + \LEthreeC - \LEthreeB - \LPthreeB + \LPthreeC + \LKthree
  && \text{by the $\tau$-labelled version of \Cref{lemma:basic-algebra}}.
\end{align*}
Here, $\LIthree$, $\LEthreeB$, $\LEthreeC$, $\LPthreeB$, $\LPthreeC$, and $\LKthree$ are the six
$\tau$-labelled graphs on vertex set $[3]$ with the obvious edge sets; as above, the distinguished
vertex is shown as an empty circle.

Therefore, by \Cref{lemma:adjoint-transfer},
\[
\alpha^\dagger_d\!\left(\LIthree + \LEthreeC - \LEthreeB - \LPthreeB + \LPthreeC + \LKthree\right) \geq 0
\]
is a valid inequality
in the original flag algebra. This inequality implies the one in  \Cref{eqn:mantel-key-inequality} because the left-hand side of the former is equal to that of the latter as shown by a series of valid equations below:
\begin{align*}
  & \alpha^\dagger_d\!\left(\LIthree + \LEthreeC - \LEthreeB - \LPthreeB + \LPthreeC + \LKthree\right)
  \\
  & \qquad {} =
  \begin{aligned}[t]
    \alpha_d\left(\LIthree\right) + \alpha_d\left(\LEthreeC\right)
    & - \alpha_d\left(\LEthreeB\right)
    - \alpha_d\left(\LPthreeB\right) 
    \\
    & {} + \alpha_d\left(\LPthreeC\right) + \alpha_d\left(\LKthree\right)
  \end{aligned}
  && \text{by the definition of } \alpha_d^\dagger
  \\
  & \qquad {} =
  \Ithree + \tfrac{1}{3} \cdot \Ethree - \tfrac{2}{3} \cdot \Ethree
  - \tfrac{2}{3} \cdot \Pthree + \tfrac{1}{3} \cdot \Pthree + \Kthree
  && \text{by the definition of } \alpha_d
  \\
  & \qquad {} = \Ithree - \tfrac{1}{3} \cdot \Ethree - \tfrac{1}{3} \cdot \Pthree + \Kthree
  && \text{by \Cref{lemma:basic-algebra}}.
\end{align*}

\section{Final Remarks}
\label{sec:conclusion}

We have presented flag algebra from the perspective of computer scientists working in logic, programming languages, automated verification, and formal methods. Our presentation emphasises that flag algebra forms a logical system for establishing asymptotic inequalities in extremal graph theory, with a well-defined syntax, semantics, and proof strategies. We described in detail one common proof strategy: proving inequalities in a labelled variant of flag algebra---often starting from manifestly valid sum-of-squares inequalities---and then transferring them to the unlabelled setting using the downward operator. We also highlighted that this transfer mechanism relies on the notion of an adjoint pair, reminiscent of Galois connections and categorical adjunctions. Just as these notions relate different interpretations of programming languages or logics, adjoint pairs relate different variants of flag algebra.

This perspective suggests several research directions for logic and formal-methods researchers. First, it would be interesting to develop a proof theory for flag algebra. The strength and limitations of the downward-operator-based strategy have been analysed in the literature. For example, when the strategy is restricted to start from the sum-of-squares inequalities in \Cref{cor:basic-algebra-tau:sos}, it is known to be incomplete~\cite{Hatami2011undecidability,Blekherman2020simple}. Similar analyses for other proof strategies in Razborov's original development~\cite{Razborov2007}---such as those based on the upward operator or graph differentiation---would be valuable. More broadly, designing proof systems for flag algebra with good proof-theoretic properties while remaining effective in practice is an appealing direction.

Second, it would be interesting to develop more effective automation and software tools for constructing proofs in flag algebra and for solving the resulting optimisation problems. As noted in the introduction, {\sc Flagmatic} has been highly successful for many extremal problems, but it largely relies on the downward-operator approach with sum-of-squares starting inequalities, and it does not scale well to medium-sized pattern graphs. Developing scalable automation that goes beyond the downward-operator/sum-of-squares paradigm is therefore an important challenge.

\section*{Acknowledgments}
We would like to thank Joonkyung Lee, Sang-il Oum, Hyunwoo Lee, Taeyoung Kim, and Jaewon Moon for helping us to understand Razborov's flag algebra. This work was supported by the National Research Foundation of Korea(NRF) grant funded by the Korean Government(MSIT) (No. RS-2023-00279680).

\bibliographystyle{plain} 
\bibliography{references}

\end{document}